\documentclass[english,draftcls,onecolumn]{IEEEtran}
\usepackage[T1]{fontenc}
\usepackage[latin9]{inputenc}
\usepackage{amsmath}
\usepackage{graphicx}
\usepackage{amssymb}
\usepackage{esint}

\newtheorem{thm}{Theorem}
  \newtheorem{lemma}[thm]{Lemma}
 \newtheorem{corollary}[thm]{Corollary}

\def\d{\mathrm{d}}
 \def\R{\mathbb{R}}
\renewcommand\[{\begin{equation}} 
\renewcommand\]{\end{equation}}

\usepackage{babel}

\begin{document}

\title{Spatial and Temporal Correlation of the Interference in ALOHA Ad
Hoc Networks}

\author{Radha Krishna Ganti and Martin Haenggi\\
Department of Electrical Engineering\\
University of Notre Dame\\
Notre Dame, IN 46556, USA\\
\{rganti,mhaenggi\}@nd.edu}
\maketitle
\begin{abstract}
Interference is a main limiting factor of the performance of a wireless
ad hoc network. The temporal and the spatial correlation of the interference
makes the outages correlated temporally (important for retransmissions) and spatially
correlated (important for routing). In this letter we quantify the temporal and
spatial correlation of the interference in a wireless ad hoc network 
whose nodes are distributed as a Poisson point process on the plane
when ALOHA is used as the  multiple-access scheme.
\end{abstract}

\section{Introduction}

Interference in a wireless ad hoc network is a spatial phenomenon
which depends on the set of transmitters, the path loss, and the fading.
The presence of common randomness in the locations of the interferers
induces temporal and spatial correlations in the interference, even
for ALOHA.  These correlations affect the retransmission strategies
and the routing. In the literature, these correlations are generally
neglected for the purpose of analytical tractability and   because
these correlations do not change the scaling behavior of an ad hoc
wireless network. For example, in \cite{bacelli-aloha} and \cite{haenggi2005rrr},
the spatial correlations are neglected for the purpose of routing.
Also extending results like the transmission capacity \cite{weber:2005}
from a single-hop to a multi-hop scenario   requires  taking  the
spatio-temporal correlations into account. In this letter we quantify
the spatial and temporal correlations of the interference and the
link outages for ALOHA.

\section{System Model}

We model the location of the nodes (radios) as a Poisson point process
(PPP) $\phi=\left\{ x_{1},x_{2},\ldots\right\} \subset\mathbb{R}^{2}$
of density $\lambda$. We assume that all the nodes transmit with
unit power and that  the fading is spatially and temporally independent
with unit mean. The (power) fading coefficient between two pairs of
nodes $x$ and $y$ at time instant $n$ is denoted by $h_{xy}(n)$.
The large scale path loss function is denoted by $g(x)$ and is assumed
to have the following properties: 
\begin{enumerate}
\item Depends only on $\|x\|$.
\item Monotonically decreases with $\|x\|$.
\item Integrable: \begin{equation}
\int_{0}^{\infty}xg(x)\d x<\infty.\label{eq:integrable}\end{equation}

\end{enumerate}
For example, a valid path loss model is given by \begin{equation}
g_{\epsilon}(x)=\frac{1}{\epsilon+\|x\|^{\alpha}},\ \epsilon\in(0,\infty),\ \alpha>2.\label{eq:11}\end{equation}
We can model the standard singular path loss model $g(x)=\|x\|^{-\alpha}$
by considering the limit $\lim_{\epsilon\rightarrow0}g_{\epsilon}(x)$.
The interference at time instant $m$ and (spatial) location $z$
is given by \begin{equation}
I_{k}(z)=\sum_{x\in\phi}\mathbf{1}(x\in\phi_{k})h_{xz}(k)g(x-z).\label{eq:22}\end{equation}
where $\phi_{k}$ denotes the transmitting set at time $k$. We assume
that the MAC protocol used is ALOHA where each node decides to transmit
independently with probability $p$ in each slot.
\section{Spatio-Temporal Correlation of Interference\label{sec:Temporal-Correlation.}}
In a wireless system the transmitting set changes at every time slot
because of the MAC scheduler. Since the transmitting sets at different
time slots are chosen from $\phi$ (a common source of randomness),
the interference exhibits temporal  and spatial correlation.  Since ALOHA
chooses the transmitting sets identically across time, $I_{k}(u)$
is identically distributed for all $k$. Since nodes transmit independently
of each other in ALOHA, the transmitting set $\phi_{k}\subset\phi$
is also spatially stationary, and hence $I_{k}(u)\stackrel{d}{=}I_{k}(o)$
where $\stackrel{d}{=}$ denotes equality in distribution and $o$
denotes the origin in $\R^2$. Hence we have 
\begin{eqnarray}
\mathbb{E}I_{k}(u) & = & \mathbb{E}I_{k}(o)\nonumber \\
 & \stackrel{(a)}{=} & \mathbb{E}\sum_{x\in\phi}\mathbf{1}(x\in\phi_{k})h_{xo}(k)g(x)\nonumber \\
 %& = & p\mathbb{E}[h]\lambda\int_{\R^{2}}g(x)\d x\nonumber \\
 & \stackrel{(b)}{=} & p\lambda\int_{\R^{2}}g(x)\d x,\label{eq:average_main}\end{eqnarray}
where $(a)$ follows from Campbell's theorem \cite{stoyan} and $(b)$
follows since $\mathbb{E}[h]=1$. The second moment of the interference
is given by \begin{eqnarray}
\mathbb{E}[I_{k}(o)^{2}] & = & \mathbb{E}\left[\left(\sum_{x\in\phi_{k}}h_{xo}(k)g(x)\right)^{2}\right]\nonumber \\
 & = & \mathbb{E}\sum_{x\in\phi_{k}}h_{xo}^{2}(k)g^{2}(x)\nonumber \\
 &  & +\mathbb{E}\sum_{x,y\in\phi_{k}}^{x\neq y}h_{xo}(k)h_{yo}(k)g(x)g(y)\nonumber 
\end{eqnarray}
\begin{eqnarray}
 & \stackrel{(a)}{=} & p\mathbb{E}[h^{2}]\lambda\int_{\R^{2}}g^{2}(x)\d x\nonumber \\
 &  & +p^{2}\mathbb{E}[h]^{2}\lambda^{2}\int_{\R^{2}}\int_{\R^{2}}g(x)g(y)\d x\d y.\label{eq:avg2}\end{eqnarray}
where $(a)$ follows from the independence of $h_{xo}(k)$ and $h_{yo}(k)$
and the second-order product density formula of the Poisson point
process \cite{stoyan}. When the fading follows a Nakagami-$m$%
\footnote{The distribution is given by \[
F(x)=1-\frac{\Gamma_{\text{ic}}(m,mx)}{\Gamma(m)},\]
where $\Gamma_{\text{ic}}$ denotes the incomplete gamma function. %
} distribution and the path loss model is  given by $g_{\epsilon}(x)$,
the variance of the interference follows from \eqref{eq:average_main}
and \eqref{eq:avg2}   and is given by \[
\text{Var}\left[I_{k}(o)\right]=\frac{2\pi^{2}(\alpha-2)p\lambda}{\epsilon^{2-2/\alpha}\alpha^{2}\sin(2\pi/\alpha)}\frac{m+1}{m},\]
 and the mean product of $I_{k}(u)$ and $I_{l}(v)$ at times $k$
and $l,\ k\neq l$ is given by \begin{eqnarray}
 &  & \mathbb{E}[I_{k}(u)I_{l}(v)]\nonumber \\
 & = & \mathbb{E}\left[\sum_{x\in\phi_{k}}h_{xu}(k)g(x-u)\sum_{y\in\phi_{l}}h_{yv}(l)g(y-v)\right]\nonumber \\
 & = & p^{2}\mathbb{E}[h]^{2}\lambda\int_{\R^{2}}g(x-u)g(x-v)\d x\nonumber \\
 &  & +\mathbb{E}\sum_{x,y\in\phi}^{x\neq y}\mathbf{1}(x\in\phi_{k})\mathbf{1}(y\in\phi_{l})h_{xu}(k)h_{yv}(l)g(x)g(y).\nonumber  
\end{eqnarray}
By Campbell's theorem and the second order product density of a PPP, we have
\begin{eqnarray}
 \mathbb{E}[I_{k}(u)I_{l}(v)]\nonumber& = & p^{2}\mathbb{E}[h]^{2}\lambda\int_{\R^{2}}g(x-u)g(x-v)\d x\nonumber \\
 &  & +\lambda^{2}p^{2}\mathbb{E}[h]^{2}\int_{\R^{2}}\int_{\R^{2}}g(x)g(y)\d x\d y\\
 & = & p^{2}\lambda\int_{\R^{2}}g(x-u)g(x-v)\d x\label{eq:cross_correlation} \\
 &  & +\lambda^{2}p^{2}\left(\int_{\R^{2}}g(x)\d x\right)^{2}.\end{eqnarray}

\begin{lemma}
\label{lem:The-spatio-temporal-correlation}The spatio-temporal correlation
coefficient of the interferences $I_k(u)$ and $I_l(v), k\neq l$,  when the path loss function $g(x)$
satisfies \eqref{eq:integrable} is given by \begin{equation}
\zeta(u,v)=\frac{p\int_{\R^{2}}g(x)g(x-\|u-v\|)\d x}{\mathbb{E}[h^{2}]\int_{\R^{2}}g^{2}(x)\d x}.\label{eq:spatial-cor}\end{equation}
\end{lemma}
\begin{proof}
Since $I_{k}(u)$ and $I_{l}(v)$ are identically distributed, we
have \[
\zeta(u,v)=\frac{\mathbb{E}[I_{k}(u)I_{l}(v)]-\mathbb{E}[I_{k}(u)]^{2}}{\mathbb{E}[I_{k}(u)^{2}]-\mathbb{E}[I_{k}(u)]^{2}}.\]
Since $I_{k}(u)\stackrel{d}{=}I_{k}(o)$ and by substituting
for the above quantities we have,\begin{eqnarray}
\zeta(u,v) & = & \frac{p\int_{\R^{2}}g(x-u)g(x-v)\d x}{\mathbb{E}[h^{2}]\int_{\R^{2}}g^{2}(x)\d x}\nonumber \\
 & \stackrel{(a)}{=} & \frac{p\int_{\R^{2}}g(x)g(x-\|u-v\|)\d x}{\mathbb{E}[h^{2}]\int_{\R^{2}}g^{2}(x)\d x},\label{eq:cross-corr}\end{eqnarray}
where $(a)$ follows by using the substitution $y=x-u$ and the fact
that $g(x)$ depends only on $\|x\|$.
\end{proof}
We have the following result about the temporal correlation by setting
$\|u-v\|=0$.
\begin{corollary}
The temporal correlation coefficient with ALOHA as the MAC protocol
and is given by \begin{equation}
\zeta_{t}=\frac{p}{\mathbb{E}[h^{2}]}.\label{eq:temporal-cor}\end{equation}
When the fading is Nakagami-$m$, the correlation coefficient is $\zeta_{t}=\frac{pm}{m+1}$. In particular, for $m=1$ (Rayleigh fading), the temporal correlation coeffecient is $p/2$ and  for $m\rightarrow \infty$ (no fading), the temporal correlation coeffecient is $p$.
\end{corollary}
We first observe that the correlation increases with increasing $m$,
i.e., fading decreases correlation which is intuitive. Observe that
in the above derivation, $\int_{\R^{2}}g^{2}(x)\d x$ is not defined
when $g(x)=\|x\|^{-\alpha}$, but we can use $g_{\epsilon}(x)$ and
take $\epsilon\rightarrow0$. We now find the correlation for the
singular path-loss model as a limit of $g_{\epsilon}(x)$.
\begin{corollary}
Let the path loss model be given by $g_{\epsilon}(x)=1/(\epsilon+\|x\|^{\alpha})$.
We then have \[
\lim_{\epsilon\rightarrow0}\zeta(u,v)=0,\ \ u\neq v.\]
\end{corollary}
\begin{proof}
We have \begin{eqnarray*}
\zeta(u,v) & = & \lim_{\epsilon\rightarrow0}\frac{p\int_{\R^{2}}g_{\epsilon}(x-u)g_{\epsilon}(x-v)\d x}{\mathbb{E}[h^{2}]\int_{\R^{2}}g_{\epsilon}^{2}(x)\d x}\\
 & \stackrel{(a)}{=} & \lim_{\epsilon\rightarrow0}\frac{p\int_{\R^{2}}\frac{1}{1+\|x-u\epsilon^{-1/\alpha}\|^{\alpha}}\frac{1}{1+\|x-v\epsilon^{-1/\alpha}\|^{\alpha}}\d x}{\mathbb{E}[h^{2}]\int_{\R^{2}}\left(\frac{1}{1+\|x\|^{\alpha}}\right)^{2}\d x}\\
 & = & 0,\end{eqnarray*}
where $(a)$ follows from change of variables.
\end{proof}
The correlation coefficient being $0$ is an artifact of the singular
path loss model. When the path loss is $\|x\|^{-\alpha}$, the nearest
transmitter is the main contributor to the interference. So for $u\neq v$,
the interference as viewed by $u$ is dominated by transmitters in
a disc $B(u,\delta),\delta>0$ of radius $\delta$ centered at $u$
and for $v$ dominated by transmitters in $B(v,\delta)$ for small
$\delta$. The transmitters locations being independent in $B(v,\delta)$
and $B(u,\delta)$ for a PPP, makes the correlation-coefficient go
to zero. A more powerful metric like mutual information would be better able to 
capture the dependence of interference for the singular path loss
model. In Figure \ref{Flo:spacial_corr}, the spatial correlation
is plotted as a function of $\|u-v\|$ for different $\epsilon$.

\begin{figure}
\begin{centering}
\includegraphics[width=0.7\columnwidth]{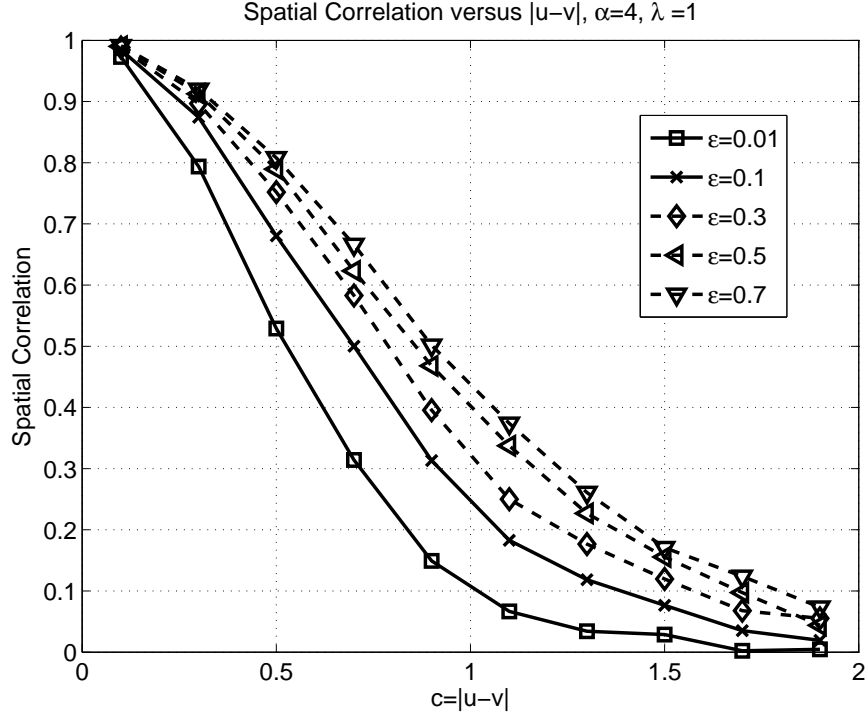}
\par\end{centering}

\caption{Spatial correlation $\zeta(u,v)/p$ versus $\|u-v\|$, when the path-loss
model is given by $g_{\epsilon}(x)$, $\lambda=1$ and $\alpha=4$.
We observe that $\zeta_{s}(u,v)\rightarrow0,\ u\neq v$, for $\epsilon\rightarrow0$.}

\label{Flo:spacial_corr}
\end{figure}

\section{Temporal Correlation of Link Outages}

In the standard analysis of retransmissions in a  wireless ad hoc system, the link failures are assumed to be  uncorrelated across time. But this is not so, since the interference is temporally
correlated. We now provide the conditional probability of link formation
assuming a successful transmission. 

We assume that a transmitter at the origin has a destination located
at $z\in\R^{2}$. Let $A_{k}$ denote the event that the origin is
able to connect to its destination $z$ at time instant $k$, i.e.,
\[
\text{SIR}=\frac{h_{oz}(k)g(z)}{I_{k}(z)}>\theta.\]
For simplicity we shall assume the fading is Rayleigh (similar methods
can be used for Nakagami-$m$). We now provide the joint probability
of success $\mathbb{P}(A_{k},A_{l}),\ k\neq l$. We have \begin{eqnarray}
\mathbb{P}(A_{k},A_{l}) & = & \mathbb{P}\left(h_{oz}(k)>aI_{k}(z),h_{oz}(l)>aI_{l}(z)\right)\nonumber \\
 & \stackrel{(a)}{=} & \mathbb{E}\left[\exp(-aI_{k}(z))\exp(-aI_{l}(z))\right]\nonumber \\
 & = & \mathbb{E}[\exp(-a\sum_{x\in\phi}g(x)[\mathbf{1}(x\in\phi_{k})h_{xz}(k)\nonumber \\
 &  & +\mathbf{1}(x\in\phi_{l})h_{xz}(l)])]\nonumber \\
 & \stackrel{(b)}{=} & \mathbb{E}\left[\prod_{x\in\phi}\left(\frac{p}{1+ag(x)}+1-p\right)^{2}\right]\label{eq:outages}\\
 & \stackrel{(c)}{=} & \exp\left(-\lambda\int_{\R^{2}}1-\left(\frac{p}{1+ag(x)}+1-p\right)^{2}\d x\right),\nonumber \end{eqnarray}
where $a=\theta/g(z)$. $(a)$ follows from the independence of $h_{oz}(k)$
and $h_{oz}(l),k\neq l$, $(b)$ follows by taking the average with
respect to $h_{xz}(k),\ h_{xz}(l)$ and the ALOHA, $(c)$ follows
from the probability generating functional of the PPP. Similarly we
have \begin{eqnarray*}
\mathbb{P}(A_{l}) & = & \exp\left(-\lambda\int_{\R^{2}}1-\left(\frac{p}{1+ag(x)}+1-p\right)\d x\right).\end{eqnarray*}
 So the ratio of conditional and the unconditional probability is
given by \begin{eqnarray}
\frac{\mathbb{P}(A_{k}|A_{l})}{\mathbb{P}(A_{l})} & = & \frac{\mathbb{P}(A_{k},A_{l})}{\mathbb{P}(A_{l})^{2}}\nonumber \\
 & = & \exp\left(\lambda p^{2}\int_{\R^{2}}\left(\frac{ag(x)}{1+ag(x)}\right)^{2}\d x\right)\nonumber \\
 & > & 1.\label{eq:ineq}\end{eqnarray}
When $g(x)=\|x\|^{-\alpha}$, we have \[
\frac{\mathbb{P}(A_{k}|A_{l})}{\mathbb{P}(A_{l})}=\exp\left(2\lambda a^{2/\alpha}p^{2}\pi^{2}\frac{(\alpha-2)}{\alpha^{2}}\csc\left(\frac{2\pi}{\alpha}\right)\right).\]
In Figure we plot the conditional and the unconditional link success
probabilities. %
\begin{figure}
\begin{centering}
\includegraphics[width=0.7\columnwidth]{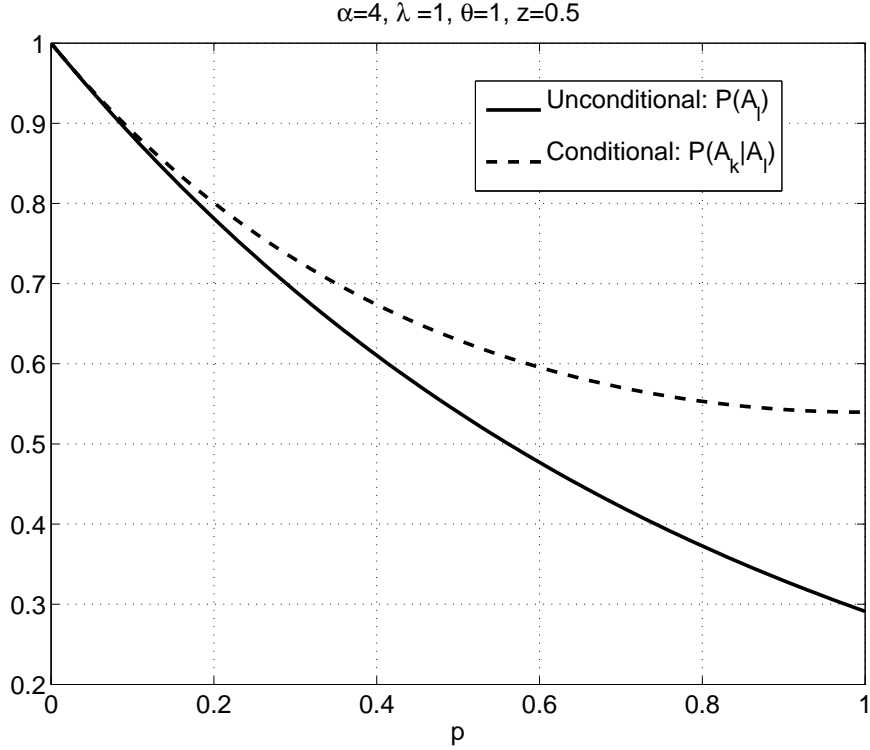}
\par\end{centering}

\caption{$\mathbb{P}(A_{k}|A_{l})$ and $\mathbb{P}(A_{l})$ versus the ALOHA
parameter $p$. $\lambda=1,$ $g(x)=\|x\|^{-4},$ $z=0.5,$ $\theta=1$. }

\end{figure}
We make the following observations:
\begin{enumerate}
\item From \eqref{eq:ineq}, we observe that the link formation is correlated
across time.
\item If a transmission succeeds at a time instant $m$, there is a higher
probability that a transmission succeeds at a time instant $n$. 
\item From \eqref{eq:ineq}, we also have $\mathbb{P}(A_{k}^{c}|A_{l}^{c})>\mathbb{P}(A_{l}^{c})$.
So a  link  in outage is always more likely to be in outage and hence
the retransmission strategy should reduce the rate of transmission
or change the density of transmitters rather than  retransmit "blindly".
\item We observe that $\frac{\mathbb{P}(A_{k}|A_{l})}{\mathbb{P}(A_{l})}$
always increases with $\theta,\lambda,p$. The increase in $\lambda$
and $p$ is because of the larger transmit set due to which the probability
of the same sub-set of nodes transmitting at different times increases,
thereby causing more correlation. When $\theta$ is large, the outage
is a result of the interfering transmissions caused by a larger number
of nodes. Hence by a similar reasoning as above, the correlation increases. 
\end{enumerate}

\section{Conclusions}

In this paper, we have derived the spatial and temporal correlations
of interference in an ALOHA wireless network. We also have proved
that the link outages are temporally correlated. This fact should
be taken into account when analyzing ad hoc performance  and  designing   retransmission strategies.
\section*{Acknowledgments}
The support of the NSF (grants CNS 04-47869, CCF 728763 ) and the
DARPA/IT-MANET program  is gratefully 
acknowledged.

\bibliographystyle{IEEEtran}
\bibliography{point_process}

\end{document}